\documentclass[]{article}
\usepackage[english]{babel}
\usepackage{amsmath,amsthm}
\usepackage{amsfonts}
\usepackage{mathrsfs}
\usepackage{color}
\usepackage{bbm}
\usepackage{latexsym}
\usepackage{graphicx}
\usepackage{enumerate}
\usepackage{float}
\usepackage{array}
\usepackage{listings}
\usepackage{subcaption}
\usepackage[font={footnotesize,it}]{caption}
\usepackage{pbox}
\newtheorem{proposition}{Proposition}[section]
\newtheorem{theorem}[proposition]{Theorem}

\newtheorem{lemma}[proposition]{Lemma}

\newtheorem{remark}[proposition]{Remark}

\numberwithin{equation}{section}
\numberwithin{proposition}{section}

\title{The effect of heterogeneity on flocking behavior and systemic risk}
\author{ Fei Fang\thanks{Mathematics and Statistics Department, Boston University(feifang@bu.edu)} , Yiwei Sun\thanks{Mathematics and Statistics Department, Boston University(yiweisun94@hotmail.com)} , and Konstantinos Spiliopoulos\thanks{Mathematics and Statistics Department, Boston University(kspiliop@math.bu.edu)} }

\begin{document}

\maketitle

\begin{abstract}
The goal of this paper is to study organized flocking behavior and systemic risk in heterogeneous mean-field interacting diffusions.  We illustrate in a number of case studies the effect of heterogeneity in the behavior of systemic risk in the system, i.e., the risk that several agents default simultaneously as a result of interconnections. We also investigate the effect of heterogeneity on the "flocking behavior" of different agents, i.e., when agents with different dynamics end up following very similar paths and follow closely the mean behavior of the system.  Using Laplace asymptotics,
we derive an asymptotic formula for the tail of the loss distribution as the number of agents grows to infinity. This characterizes the tail of the loss distribution and the effect of the heterogeneity of the network on the tail loss probability.
\end{abstract}

\section{Introduction}
Systemic risk is the risk that a large number of components of an interconnected system fail within a short time thus leading to the overall failure of the  system, see for example \cite{SystemicRiskBook,GarnierPapanicoalouYang2013b,Spiliopoulos2015}. Flocking behavior refers to the phenomenon where all or most of the components of a system follow more or less the same behavior, see for example \cite{MotschTadmor2011}.  The two phenomena, systemic risk and flocking behavior, are of course different and describe different but, at the same time, related issues. The aim of this paper is to investigate the effect of heterogeneity on  systemic risk and flocking behavior in a system of interacting diffusions.

In financial markets, for example, the combined effect of systemic risk and flocking may mean that the majority of components in a pool goes to a default state.  In reliability, a large system of interacting components might have a
central connection, crucial for normal operations of the system. The
failure of an individual component increases the stress on the central connection and thus on the other
components as well, making the entire system more likely to fail. In insurance, the system could represent
a pool of insurance policies.

We consider a stylized model of mean-field interacting diffusions:
\begin{equation}
dY_t^{(i)}=\alpha_i\left(\frac{1}{N}\sum_{j=1}^NY_t^{(j)}-Y_t^{(i)}\right)dt+\sigma_i dW_t^{(i)},i=1,...,N \label{Eq:MainModel}
\end{equation}

Model (\ref{Eq:MainModel}) is a simple enough model for mathematical analysis, yet able to capture some of the main features of the effect of heterogeneity on systemic risk and flocking behavior. Our goal in this paper is not to offer a sufficiently sophisticated model, but rather to explain the phenomena that can be expected within a reasonably rich but simple model. To induce heterogeneity, we assume that not all $\alpha_i$'s and/or not all $\sigma_i$'s are equal. The interaction of different agents, i.e., the different components of the system, is via the mean-field empirical average $\frac{1}{N}\sum\limits_{j=1}^NY_t^{(j)}$. Each $W_t^{(i)}, i=1,...,N$ represents independent standard Brownian motion. The diffusion coefficient $\sigma_{i}$'s represent the volatility  corresponding to agent $i$, and its size represent how stable the performance is for agent $i$. The coefficient $\alpha_{i}$ represent the level of influence from the whole system on agent $i$. The dynamics (\ref{Eq:MainModel}) imply that agent $i$ is attracted towards the level $\bar{Y}_{t}= \frac{1}{N}\sum_{j=1}^NY_t^{(j)}$ at time $t$.

For example, in a banking financial system, $Y_t^{(i)}$ could represent the log-monetary reserve of agent $i$ at time $t$, see for example \cite{FouqueIchiba2013}.
The paper \cite{FouqueSun2013} studies the system (\ref{Eq:MainModel}) in the homogeneous case, i.e., when $\alpha_i=\alpha$ and $\sigma_i=\sigma$ for every $i=1,...,N$. They find that along increasing $\alpha$'s, the behavior of each $Y_t^{(i)}$ is closer to the behavior of the mean $\bar{Y}_{t}=\frac{1}{N}\sum\limits_{j=1}^NY_t^{(j)}$. This phenomenon leads to the so-called "flocking effect", where all trajectories follow the same behavior in path space.
The main conclusion of \cite{FouqueSun2013} is that large $\alpha$ enhances the "flocking effect", which may either lead to a greater stability of the system or may increase the systemic risk (i.e., increase the likelihood that a large number of components in the system fail).

In our paper, we consider the heterogeneous case where not all $\alpha_i$'s or $\sigma_i$'s are equal to each other. For a fixed default level $\eta<0$, we use Laplace asymptotics to characterize the probability of the default event of the overall system $A:=\left\{\displaystyle \min_{0 \leq t \leq T} \frac{1}{N}\sum_{i=1}^NY_t^{(i)}\leq\eta\right\}$. We find that in the limit as $N\to\infty$, this probability is of the order of $e^{-\frac{\eta^2 N}{2V_T^2}}$, where $V_T^2$ depends on $\alpha_i$'s and $\sigma_i$'s in a concrete way. Of courses in the case of $\alpha_i=\alpha$ and $\sigma_i=\sigma$, we recover the formula $V_T^2=T\sigma^2$ as in \cite{FouqueSun2013}. In addition, we obtain in certain cases of interest non-asymptotic bounds for the probability $P\left(\sup_{0\leq t\leq T}|Y_{t}^{(i)}-\bar{Y}_{t}|>\delta\right)$, demonstrating the effect of $\alpha$'s and $\sigma$'s on the flocking behavior.
In addition, we  explore numerically the behavior of $Y_t^{(i)}$ and  the tail loss probability of the event $A$
 using Monte-Carlo simulation.

 We are in particularly interested in the effect of heterogeneity and derive a number of interesting conclusions (explained in detail in Sections 2 and 3):
\begin{enumerate}[(i)]
\item{If $\alpha_i=\alpha$  for every  $i$, then the larger $\sigma_i$'s are, the more volatile the system is. As $\alpha$ gets larger, the "flocking effect" becomes more apparent, but larger $\sigma_i$'s naturally lead to  deviation from the mean behavior.}
\item{If $\alpha_i=\alpha$  for every $i$, then $V_T^2=T\sigma_*^2=T \displaystyle \lim_{N\to\infty} \frac{1}{N} \sum_{i=1}^{N}\sigma_i^2$, which means that the effect of heterogeneity in the behavior of the system is described via the effective averaged diffusion $\sigma_*^2$. This implies that larger $\sigma_*^2$ leads to larger probability of losses.}
\item{In the case of $\alpha_i=\alpha$  for every $i$, we obtain an exact bound for the probability $P\left(\sup_{0\leq t\leq T}|Y_{t}^{(i)}-\bar{Y}_{t}|>\delta\right)$, which demonstrates that flocking behavior is controlled by a term of the form  $\frac{f(\sigma_{1}^{2},\cdots,\sigma_{N}^{2})}{\alpha}$. This explains why large values of $\alpha$ benefit strong flocking behavior and it also shows the effect of the volatilities, as the function $f$ takes an explicit form.}
\item{In the case of different $\alpha_i$'s and different $\sigma_i$'s, we find that the system is more stable when the number of intermediate size agents is large enough, i.e., the number of agents with moderate values of $\alpha_i$ and $\sigma_i$ is much larger than the number of agents that have either small $\alpha_i$ or large $\sigma_i$. This can be thought of as an effect of the heterogeneity of the network (a complete graph in the present case) that the agents constitute.}
\item{At the same time, we find that the probability of all agents defaulting is smaller in systems whose agents that have large $\alpha_i$ also have relatively large $\sigma_i$ and correspondingly agents with small $\alpha_i$ also have small volatility $\sigma_i$. We also analyze the form of $V^{2}_{T}$ theoretically, which illustrates the effect of the interaction of $\alpha_i$ and $\sigma_i$.}
\item{Even seemingly small changes in the composition of the structure of the population that the agents constitute, can have dramatic consequences as far as organized flocking behavior and stability of the system are concerned.}
\end{enumerate}

Systemic risk in the financial system is a subject that has grown a lot in recent years. We refer the reader to \cite{SystemicRiskBook} for a general introduction to the subject and to
\cite{FouqueIchiba2013,GarnierPapanicoalouYang2013,GarnierPapanicoalouYang2013b,GieseckeSpiliopoulosSowers2013,Spiliopoulos2015,SpiliopoulosSowers2015} for studies on clustering, rare events  and large deviations associated to systemic risk in heterogeneous financial networks. In this paper we consider a stylized model. Our main goal is to study the effect of heterogeneity and of the structure of the system on flocking behavior and on systemic risk. We consider a simple model that is amenable to analysis, but at the same time is rich enough to capture different type of interesting behaviors.

The rest of this paper is organized as follows.  In Section 2, we investigate the case of $\alpha_i=\alpha $ for every $i$ but $\sigma_i \neq \sigma_j$ for some $i \neq j$. We study the effect of heterogeneity on the "flocking effect" of $Y_t^{(i)}$
and on the tail of the default probability $\frac{1}{N}\sum\limits_{i=1}^NY_t^{(i)}$. In Section 3, we study the richer cases of $\alpha_i \neq \alpha_j$ and $\sigma_i \neq \sigma_j$ for some $i \neq j$. We explore the situation both numerically and theoretically, characterizing $V_T^2$, which controls the tail of the average loss distribution; this is Theorem \ref{T:Main}. In this general case, the formula of $V_T^2$ is complex even though it is in explicit form. To get a better understanding, we derive an approximation of $V_T^2$ via Taylor expansion in the special case of $\alpha_i = \bar{\alpha} (1+\delta c_i),  0 < \delta \ll 1$ in terms of $\delta$, this is Lemma \ref{L:Main}. Of course when $\delta = 0$, we get back the simpler formula derived in Section 2. We conclude in Section 4 with conclusions and future work.

\section{Identical $\alpha$'s but Different $\sigma$'s}
We start by exploring how different volatilities affect the default probability of the $N$ agents. Thus, we assume that $\sigma_i, i=1,...,N$ for each agent $i$ are not necessarily identical, but we assume that they all share the same mean reversion parameter $\alpha$. Thus, for all $i=1,...,N$ equation (\ref{Eq:MainModel}) takes the form:
\begin{equation}
dY_t^{(i)}= \alpha\left(\frac{1}{N}\sum_{j=1}^NY_t^{(j)}-Y_t^{(i)}\right)dt + \sigma_idW_t^{(i)}\label{Eq:Model3}
\end{equation}

To understand the behavior of the tail default probability and of flocking, we first focus on the behavior of ensemble average that reaches the default state.
Based on equation (\ref{Eq:Model3}) and with the non-important assumption $Y_0^{(i)}=0$ , $i=1,...,N$ , we get:
\begin{equation*}
Y_t^{(i)}=\alpha\int_{0}^{t}\left(\frac{1}{N}\sum_{j=1}^NY_s^{(j)}-Y_s^{(i)}\right)ds+\sigma_iW_t^{(i)}
\end{equation*}
Hence we obtain
\begin{equation}
\sum_{i=1}^N Y_t^{(i)}=\alpha\int_{0}^{t}\left(\sum_{j=1}^NY_s^{(j)}-\sum_{i=1}^NY_s^{(i)}\right)ds+\sum_{i=1}^N\sigma_iW_t^{(i)}
=\sum_{i=1}^N \sigma_i W_t^{(i)}\label{Eq:MeanBehaviorSameAlpha}
\end{equation}
Based on this relationship, we obtain the following by using the reflection principle of Brownian motion:

\begin{equation*} \begin{split}
&P\left( \min_{0 \leq t \leq T} \frac{1}{N} \sum_{i=1}^N Y_t^{(i)} \leq \eta \right)
= P\left( \min_{0 \leq t \leq T} \frac{1}{N} \sum_{i=1}^N \sigma_{i}W_t^{(i)} \leq \eta \right) \\
&\qquad= 2P \left( \frac{\sqrt{\sum_{i=1}^N \sigma_i^2}}{N}  \tilde {W_T} \leq \eta \right)
= 2 \Phi \left( \frac{N \eta}{\sqrt{T} \sqrt{\sum_{i=1}^N \sigma_i^2}} \right)
\end{split}
\end{equation*}
where $\tilde {W_t}$ denotes a standard Brownian motion at time $t$. Using Laplace asymptotics we get:\\
\begin{eqnarray*} \begin{split}
P \left( \min_{0 \leq t \leq T} \frac{1}{N} \sum_{i=1}^N Y_t^{(i)} \leq \eta \right)
\approx  2\exp\left\{-\frac{N^{2}\eta^{2}}{2T (\sum_{i=1}^N \sigma_i^{2})}\right\}
\end{split}
\end{eqnarray*}
with $\eta < 0$. Thus, we have arrived at the following theorem
\begin{theorem}\label{T:main_same_alpha}
Consider the diffusion processes $Y_{t}^{(i)}$ for $i=1,...,N$ as given by (\ref{Eq:Model3}). Let $V_{T}^{2}=T\sigma_*^{2}$, where $\sigma_*^{2} = \lim\limits_{N\to\infty} \frac{1}{N}  \sum_{i=1}^N \sigma_i^{2}$. Assume that $\sigma_*^{2}$ exists and is finite.  Then, we have that
\begin{equation} \begin{split}
\lim_{N\to\infty} -\frac{1}{N} \log P \left( \min_{0 \leq t \leq T} \frac{1}{N} \sum_{i=1}^N Y_t^{(i)} \leq \eta \right)
& = \frac{\eta^{2}}{2V_{T}^{2}}.
\end{split}
\end{equation}
\end{theorem}

Hence, if N is large enough, the following approximation holds:
\begin{equation}
P \left( \min_{0 \leq t \leq T} \frac{1}{N} \sum_{i=1}^N Y_t^{(i)} \leq \eta \right) \approx e^{-\frac{N \eta^{2}}{2T\sigma_*^{2}}}\label{Eq:ApproximationSameAlpha}
\end{equation}

From the equation above, if $\sigma_*^{2}$ increases, the probability that the empirical mean $\frac{1}{N} \sum_{i=1}^N Y_t^{(i)}$ falls below $\eta$ for some $t \in [0,T]$ will increase. This verifies the intuition that the tail default probability increases as a function of $\frac{1}{N}\sum_{i=1}^N \sigma_i^{2}$.

Next, we turn to flocking behavior. In order to understand the flocking behavior quantitatively, we need, for a given $\delta>0$, to control the probability $P\left(
\sup_{0\leq t\leq T}|Y_{t}^{(i)}-\bar{Y}_{t}|>\delta\right)$.

\begin{proposition}\label{P:Flocking}
 Consider the diffusion processes $Y_{t}^{(i)}$ for $i=1,...,N$ as given by (\ref{Eq:Model3}). Let $\kappa_{i}=\sqrt{(1-1/N)^{2}\sigma_i^{2}+(1/N^{2})\sum_{j\neq i}\sigma^{2}_{j}}$. Then, for any $\delta>0$, we have that
 \begin{align}
P\left(\sup_{0\leq t\leq T}|Y_{t}^{(i)}-\bar{Y}_{t}|>\delta\right)&\leq
2\exp\left\{-\frac{\delta^2}{\frac{\kappa_{i}^{2}}{\alpha}(1-e^{-2\alpha T})}\right\}
 \end{align}
\end{proposition}

In particular, Proposition \ref{P:Flocking} shows that flocking behavior is controlled by the "flocking" parameter
\[
F=\max_{i=1,\cdots, N}\frac{\kappa_{i}^{2}}{\alpha}=\frac{(1-1/N)^2\underset{i=1,\cdots, N}{\max} \sigma_{i}^{2}+\frac{1}{N^{2}}\underset{j\neq i}{\sum}\sigma_{j}^{2}}{\alpha}
\]
which, shows that if $\alpha$ is large and $\sigma_{i}'$s are small, then the particles closely follow in probability the behavior of the mean.

\begin{proof}[Proof of Proposition \ref{P:Flocking}]
Based on (\ref{Eq:MeanBehaviorSameAlpha}), we have that  $R^{(i)}_{t}=Y^{(i)}_{t}-\bar{Y}_{t}$ satisfies
\begin{align*}
R_t^{(i)}&=-\alpha\int_{0}^{t}R_s^{(i)}ds+\sigma_iW_t^{(i)}-\frac{1}{N}\sum_{i=1}^{N}\sigma_{i}W^{(i)}_{t}
\end{align*}

Due to the independence of the Wiener processes, $W_{t}^{(i)}, i=1,..., N$, we obtain that there exists a Brownian motion $\tilde{W}^{(i,N)}_{t}$ such that, in distribution
\begin{align*}
R_t^{(i)}&=-\alpha\int_{0}^{t}R_s^{(i)}ds+\kappa_i\tilde{W}_t^{(i,N)}
\end{align*}
 Solving this SDE, we then obtain
\begin{align*}
R_t^{(i)}&=\kappa_{i}\int_{0}^{t}e^{-\alpha(t-s)}d\tilde{W}_s^{(i,N)}\sim N\left(0,\frac{\kappa_{i}^{2}}{2\alpha}(1-e^{-2\alpha t})\right)
\end{align*}
Let $\lambda > 0$ be given, then we have, using the submartingale property of $e^{\lambda R_{t}^{(i)}}$:
\begin{align*}
P\left(\sup_{0\leq t\leq T}|R_{t}^{(i)}|>\delta\right) \leq 2P\left(\sup_{0\leq t\leq T}R_{t}^{(i)}>\delta\right)
&= 2P\left(\sup_{0\leq t\leq T}e^{\lambda R_{t}^{(i)}} > e^{\lambda \delta}\right) \\
&\leq 2e^{-\delta \lambda}E[e^{\lambda R_{T}^{(i)}}] \\
&\leq 2e^{-\delta \lambda + \frac{\lambda^2}{2}\frac{\kappa_{i}^{2}}{2\alpha}(1-e^{-2\alpha T})}
\end{align*}
Minimizing over $\lambda$ we get the bound
\begin{align*}
P\left(\sup_{0\leq t\leq T}|R_{t}^{(i)}|>\delta\right) \leq 2e^{-\delta^2(\frac{\kappa_{i}^{2}}{\alpha}\left(1-e^{-2\alpha T})\right)^{-1}}
\end{align*}
concluding the proof of the Proposition.
\end{proof}

\section{Different $\alpha$'s and Different $\sigma$'s}
In this section, we study how the combination of different interactions and different volatilities affect flocking behavior and the tail default probability distribution. We also study the effect of the structure of the system on its stability.

\subsection{Numerical simulations and motivation}

 To motivate the discussion, we first solve numerically (\ref{Eq:MainModel}) using the standard  Euler  scheme. We  plot two groups of trajectories of $N$ agents to compare two extreme situations: (i) Group A: agents with small $\alpha_i$'s  have large $\sigma_i$'s while agents with large $\alpha_i$'s have small $\sigma_i$'s; (ii) Group B: agents with small $\alpha_i$'s also have  small $\sigma_i$'s and agents with large $\alpha_i$'s also have large $\sigma_i$'s. That is, we design two groups as follows:
 \begin{align*}
\text{Group A}&:  \{(\alpha,\sigma)_{\{1,2\}},(\alpha,\sigma)_{\{3,4,5,6,7\}},(\alpha,\sigma)_{\{8,9,10\}}\}=\{(1,2),(10,1),(100,0.5)\}\nonumber\\
\text{Group B}&:\{(\alpha,\sigma)_{\{1,2\}},(\alpha,\sigma)_{\{3,4,5,6,7\}},(\alpha,\sigma)_{\{8,9,10\}}\}=\{(1,0.5),(10,1),(100,2)\}.\nonumber
 \end{align*}

Sample trajectories and loss distribution are shown on Figure \ref{fig:4.1}, Figure \ref{fig:4.2} and Figure \ref{fig:4.3}, Figure \ref{fig:4.4} respectively. In both trajectories graphs, two light grey lines are trajectories with $(\alpha,\sigma)_{\{1,2\}}$; five dark grey lines are trajectories with $(\alpha,\sigma)_{\{3,4,5,6,7\}}$; and three black lines are trajectories with $(\alpha,\sigma)_{\{8,9,10\}}$. The solid horizontal line represents the default level $\eta=-0.7$ and the centered bold red line represents the average trajectory.

We shall refer to tail default probability to be the probability that all agents default.
\begin{figure}[H]
\centering
\begin{minipage}{.49\textwidth}
\centering
\includegraphics[width=\linewidth]{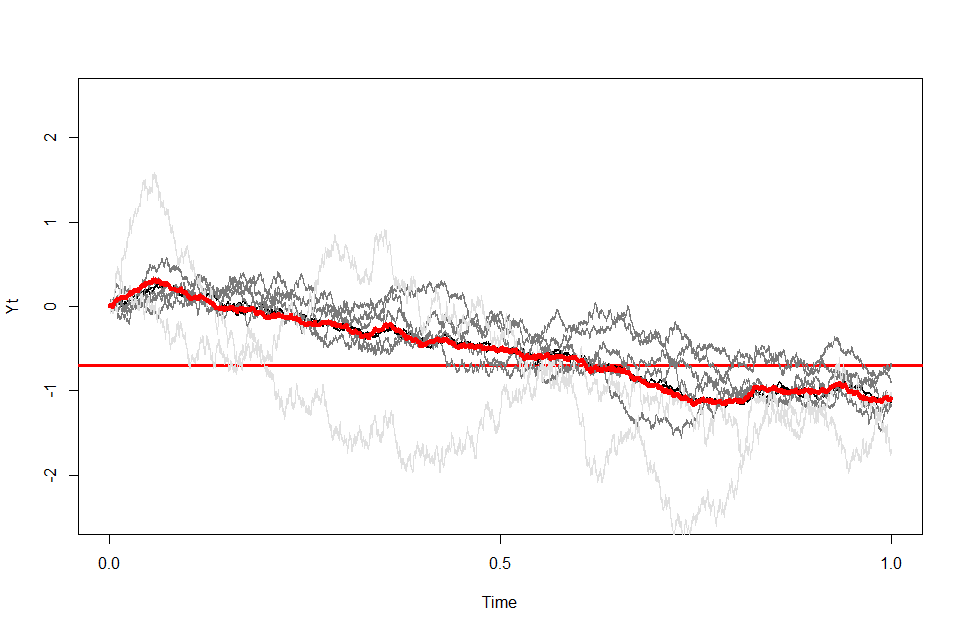}
\caption{Trajectories for Group A $\{(1,2),(10,1),(100,0.5)\}$}
\label{fig:4.1}
\end{minipage}\hfill
\begin{minipage}{.49\textwidth}
\centering
\includegraphics[width=\linewidth]{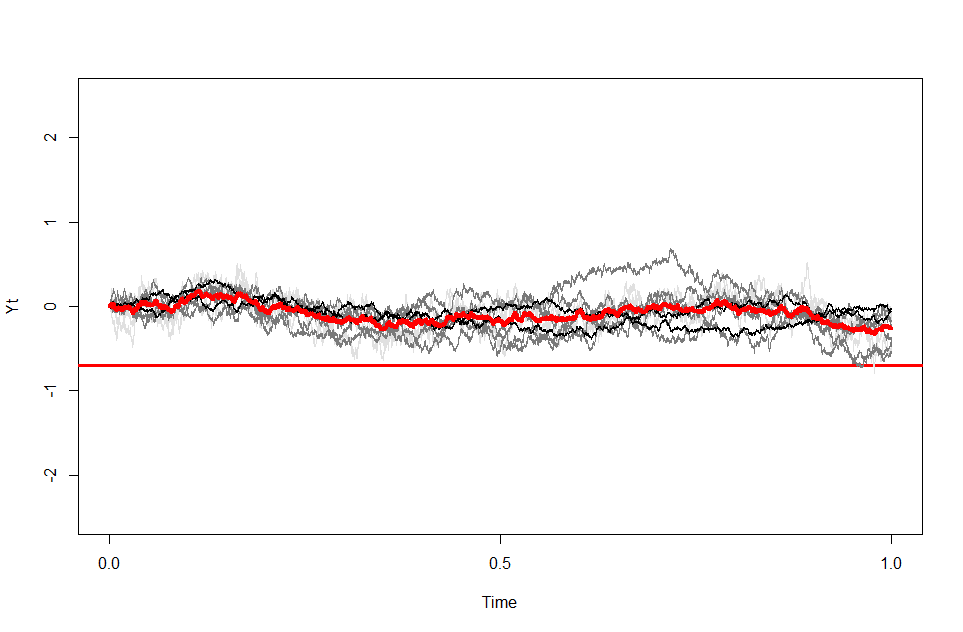}
\caption{Trajectories for Group B $\{(1,0.5),(10,1),(100,2)\}$}
\label{fig:4.3}
\end{minipage}
\end{figure}

\begin{figure}[H]
\centering
\begin{minipage}{.49\textwidth}
\centering
\includegraphics[width=\linewidth]{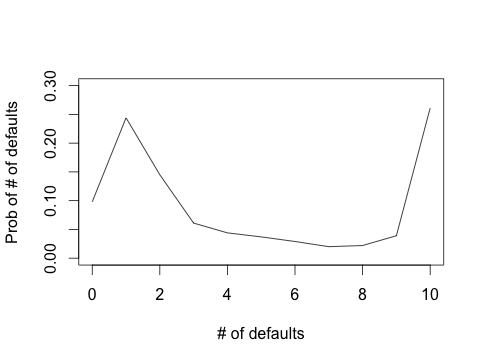}
\caption{Loss distribution with tail default probability 0.261}
\label{fig:4.2}
\end{minipage}\hfill
\begin{minipage}{.49\textwidth}
\centering
\includegraphics[width=\linewidth]{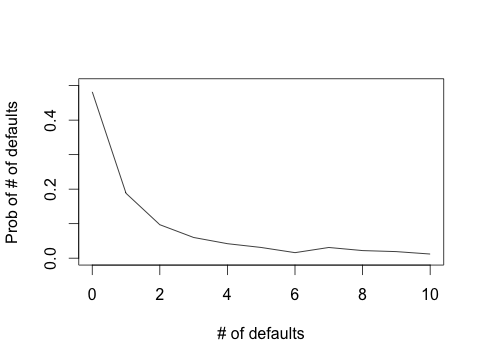}
\caption{Loss distribution with tail default probability 0.012}
\label{fig:4.4}
\end{minipage}
\end{figure}

\begin{remark}\label{R:NetworkEffect}
It is interesting to note that this seemingly small change in the structure of the system has such a profound effect on its behavior and stability.
\end{remark}

Moreover, we have the following observations 
\begin{itemize}
\item
Observation 1:
From Figures \ref{fig:4.1} and \ref{fig:4.2}, we see that agents that have: (i) small $\alpha_i$'s, (ii) large $\sigma_i$'s,  and (iii) small proportion of total number of such agents, do not  follow the mean behavior closely. However, their large volatilities contribute significantly to the trend of the overall system (or, average behavior). Other agents' trajectories seem to follow the average behavior. Meanwhile, the probability of all $N$ agents defaulting is high.

\item
Observation 2:
From Figures \ref{fig:4.3} and \ref{fig:4.4}, we see that agents that have: (i) large $\alpha_i$'s, and (ii) large $\sigma_i$'s follow the average behavior more closely, i.e., the ``flocking effect'' is more apparent. Furthermore, the probability of a large default phenomenon to happen is relatively small.
\end{itemize}

Now we introduce the concept of {\em stabilization} which will be used in the following part. Stabilization or stable behavior refers to the level of difficulty for the system to go to a default state. That is, the more stable the system is, the more difficult it is for the agents in the system to fail.

From \cite{FouqueSun2013}, we know that in the homogeneous case and under reasonable value of $\sigma$'s, $\alpha$'s dominate the stability of the system only when $\alpha$'s are large enough to cause a "flocking behavior". In this paper we see that under reasonable range of $\sigma_i$'s, when $\alpha_i$'s are relatively large to $\sigma_i$'s, the diffusion processes which have large $\alpha_i$'s will produce "flocking effect" around the average behavior. On the other hand, those that have large values of $\sigma_i$'s but small value of $\alpha_i$'s will largely affect the stability of the system. Thus, the whole system is easier to go to a default state, which is indicated by Observation 1. On the contrary, if small $\alpha_i$'s are combined with small $\sigma_{i}$'s and large $\alpha_i$'s are combined with large $\sigma_{i}$'s, then the system is more stable, which is Observation 2.

From Section 2, we learn that the tail default probability is positively correlated with $\frac{1}{N}\sum_{i=1}^N \sigma_i^2$ when $\sigma_i$'s are different but $\alpha_i$'s are the same. Now, we take different $\sigma_i$'s and $\alpha_i$'s which may lead to more complicated behavior. Therefore, we also conjecture that the composition of the system might also affect the systemic risk. 




Recall the two groups A and B, each one composed by three different subgroups:
\begin{align}
\text{Group A}&: \{(\alpha,\sigma)_{I},(\alpha,\sigma)_{II},(\alpha,\sigma)_{III}\}=\{(1,2),(10,1),(100,0.5)\} \nonumber\\
\text{Group B}&: \{(\alpha,\sigma)_{I},(\alpha,\sigma)_{II},(\alpha,\sigma)_{III}\}=\{(1,0.5),(10,1),(100,2)\}\label{Eq:DifferentSubgroups}
\end{align}

Consider two cases, Case A and Case B,  for the ratios of agents of each type within each group. Our goal is to investigate  the effect of different structures of the network on the tail default probability. In Case A, the ratios for agents of type I, II and III  are  8:1:1, 1:8:1 and 1:1:8 respectively. In Case B, the corresponding ratios  are 5:3:2, 2:5:3, 2:3:5.

In the following trajectories figures for Case A,  agents of type I, II and III  correspond to light grey lines, dark grey lines and black lines respectively. The solid horizontal line represents the "default level" $\eta=-0.7$ and the centered bold red line represents the average trajectory.
\begin{figure}[H]
\centering
{\textit{Group A: $\{(\alpha,\sigma)_{I},(\alpha,\sigma)_{II},(\alpha,\sigma)_{III}=\{(1,2),(10,1),(100,0.5)\}$}}\\
\begin{minipage}{.3\textwidth}
\centering
\includegraphics[width=\linewidth]{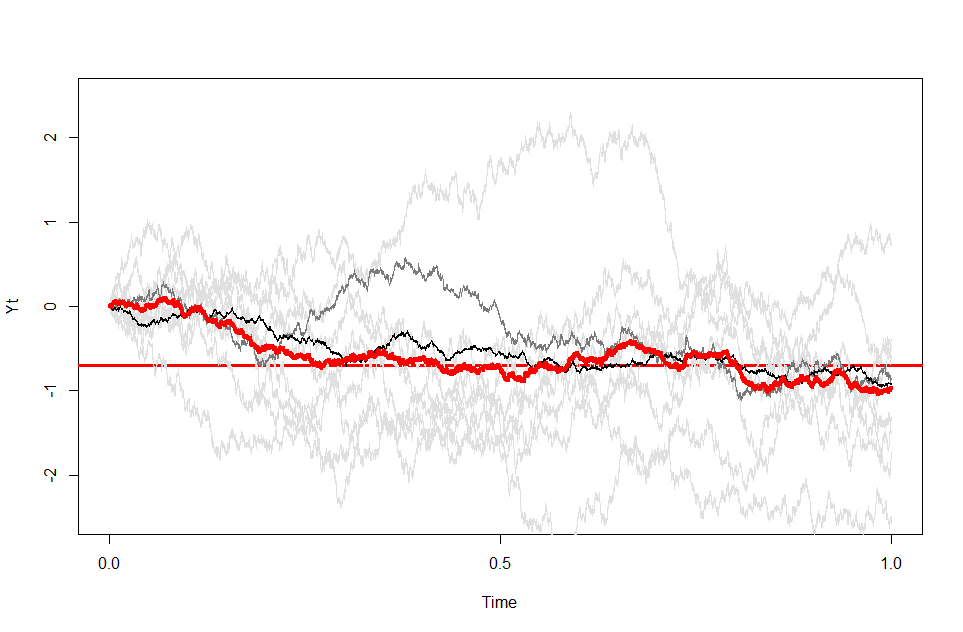}
\caption{Trajectories for Ratio of Agents 8:1:1}
\label{fig:4.5}
\end{minipage}\hfill
\begin{minipage}{.3\textwidth}
\centering
\includegraphics[width=\linewidth]{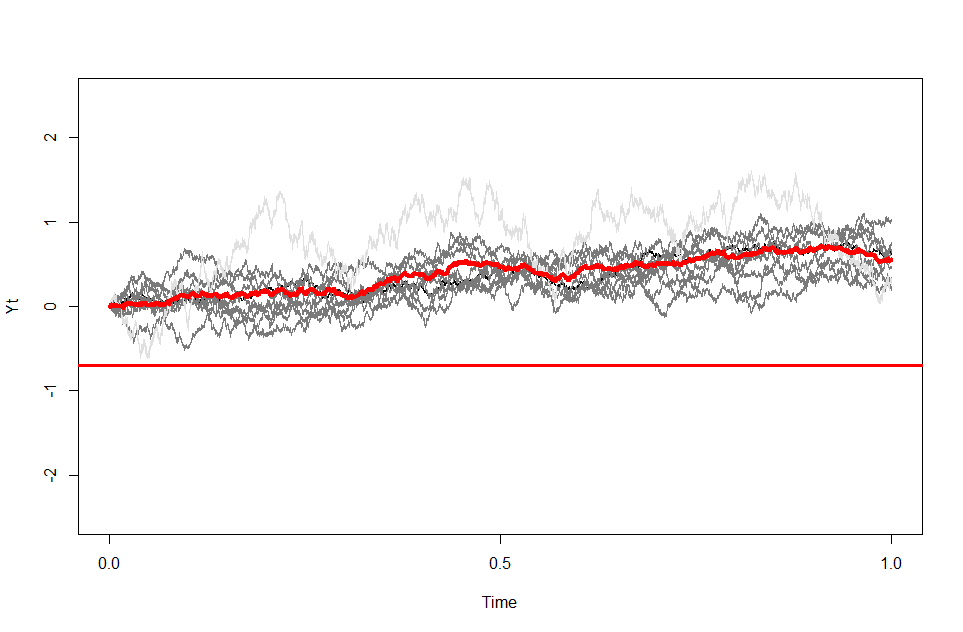}
\caption{Trajectories for Ratio of Agents 1:8:1}
\label{fig:4.6}
\end{minipage}\hfill
\begin{minipage}{.3\textwidth}
\centering
\includegraphics[width=\linewidth]{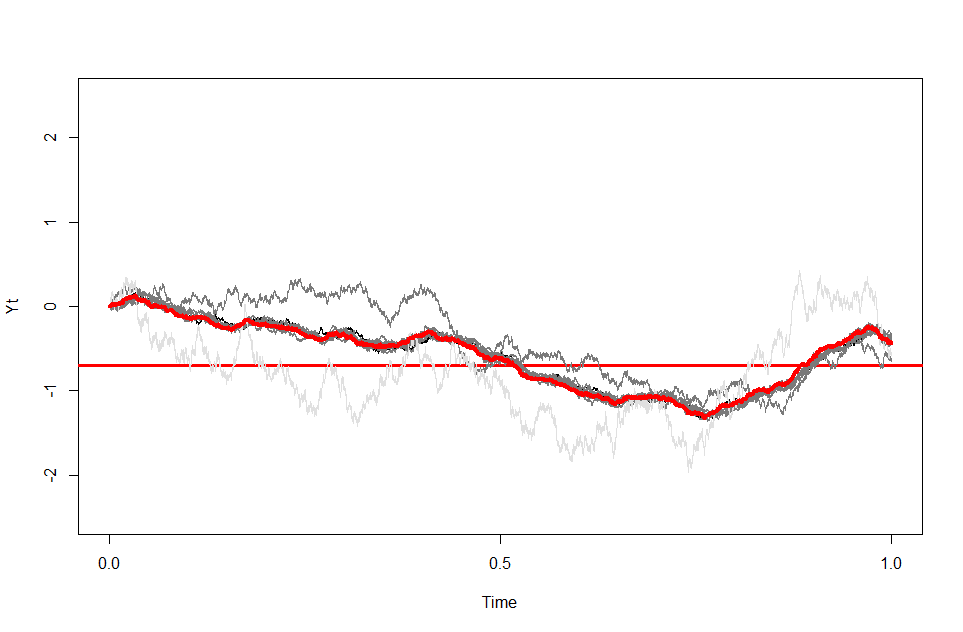}
\caption{Trajectories for Ratio of Agents 1:1:8}
\label{fig:4.7}
\end{minipage}
\end{figure}

\begin{table}[H]
\caption {Loss Distribution of Group A for Different Proportion of Agents} \label{tab:4.1}
\begin{center}
\begin{tabular}{>{\centering}m{4cm}|c|c|c|c|c|c}
\hline
number of defaulted agents of type\\ (I: II: III) & 0 & 6 & 7 & 8 & 9 & 10 \\
\hline
$(8:1:1)$ & 0.000 & 0.192 & 0.168 & 0.117 & 0.114 & 0.078\\
\hline
$(1:8:1)$ & 0.204 & 0.036 & 0.029 & 0.032 & 0.046 &0.203\\
\hline
$(1:1:8)$ & 0.274 & 0.004 & 0.005 & 0.004 & 0.019 & 0.405\\
\hline\hline
$(5:3:2)$ & 0.001 & 0.058 & 0.054 & 0.029 & 0.118 & 0.152\\
\hline
$(2:5:3)$ & 0.098 & 0.029 & 0.020 & 0.022 & 0.039 & 0.261\\
\hline
$(2:3:5)$ & 0.074 & 0.007 & 0.010 & 0.016 & 0.039 & 0.296\\
\hline
\end{tabular}
\end{center}
\end{table}
From these graphs and Table \ref{tab:4.1}, we obtain:
\begin{itemize}
\item
Observation 3:
From Figure \ref{fig:4.5} and Table \ref{tab:4.1}, we see that the system which possesses agents of subgroup I with the largest proportion seems unstable but the tail default probability is the smallest compared to other groups.
\item
Observation 4:
From Figure \ref{fig:4.6} and Table \ref{tab:4.1}, we see that  when agents of subgroup II dominate the system, the system seems to be more stable with strong "flocking behavior" and less trajectories will reach the tail default level. Furthermore, the tail default probability is intermediate in all groups.
\item
Observation 5:
From Figure \ref{fig:4.7} and Table \ref{tab:4.1}, we see that  when agents of subgroup III dominate the system, then a strong "flocking effect" exists and tail default probability is the largest when compared to the two other groups.
\end{itemize}

Next, we also explore what happens when we change the composition of the network. In particular we explore Group B, given by (\ref{Eq:DifferentSubgroups}). Note that by Observations 1 and 2, we expect that in this case the overall default probability will be smaller than that of Group A.
\begin{figure}[H]
\centering
{\textit{Group B: $\{(\alpha,\sigma)_{I},(\alpha,\sigma)_{II},(\alpha,\sigma)_{III}=\{(1,0.5),(10,1),(100,2)\}$}}\\
\begin{minipage}{.3\textwidth}
\centering
\includegraphics[width=\linewidth]{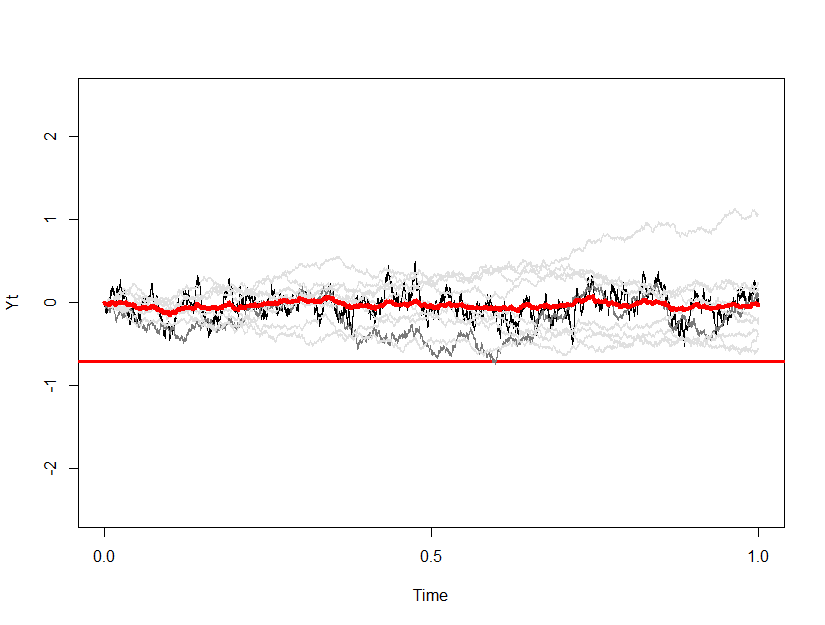}
\caption{Trajectories for Ratio of Agents 8:1:1}
\label{fig:4.8}
\end{minipage}\hfill
\begin{minipage}{.3\textwidth}
\centering
\includegraphics[width=\linewidth]{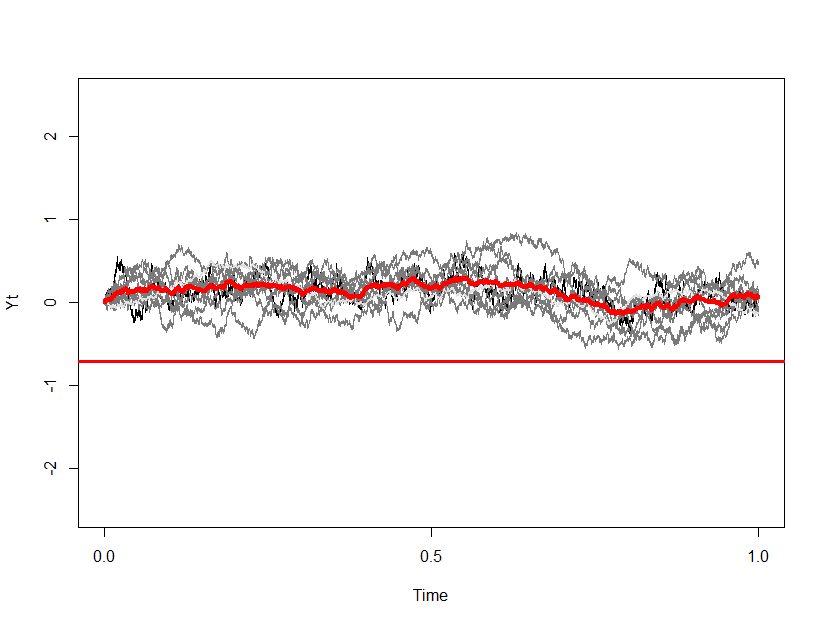}
\caption{Trajectories for Ratio of Agents 1:8:1}
\label{fig:4.9}
\end{minipage}\hfill
\begin{minipage}{.3\textwidth}
\centering
\includegraphics[width=\linewidth]{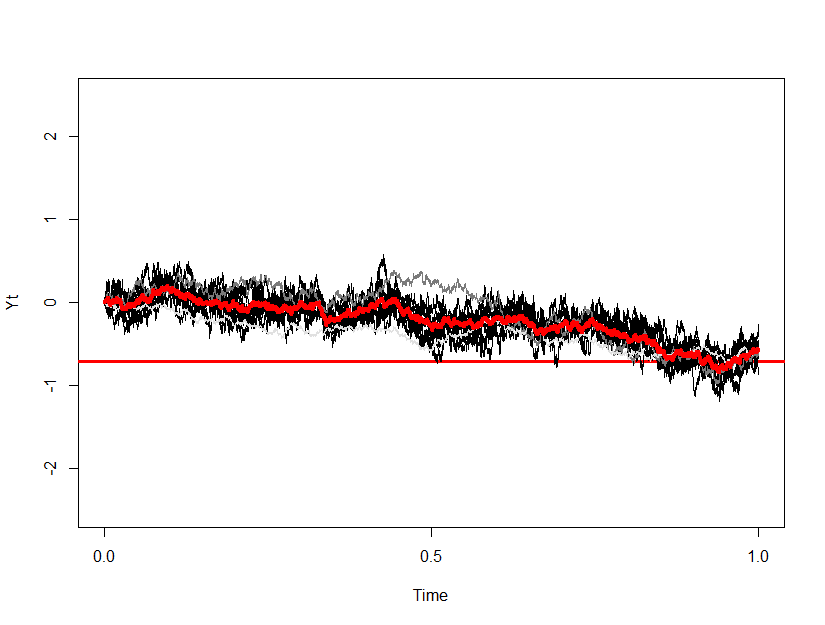}
\caption{Trajectories for Ratio of Agents 1:1:8}
\label{fig:4.10}
\end{minipage}
\end{figure}

\begin{table}[H]
\caption {Loss Distribution of Group B for Different Proportion of Agents} \label{tab:4.3}
\begin{center}
\begin{tabular}{>{\centering}m{4cm}|c|c|c|c|c|c}
\hline
number of defaulted agents of type\\ (I: II: III) & 0 & 6 & 7 & 8 & 9 & 10 \\
\hline
$(8:1:1)$ & 0.524 & 0.006 & 0.001 & 0.001 & 0.000 & 0.000\\
\hline
$(1:8:1)$ & 0.460 & 0.035 & 0.032 & 0.031 & 0.023 & 0.016\\
\hline
$(1:1:8)$ & 0.483 & 0.029 & 0.019 & 0.031 & 0.106 & 0.096\\
\hline\hline
$(5:3:2)$ &  0.510 & 0.010 & 0.008 & 0.009 & 0.001 & 0.001\\
\hline
$(2:5:3)$ &  0.489 & 0.024 & 0.024 & 0.028 & 0.032 & 0.008\\
\hline
$(2:3:5)$ & 0.516 & 0.027 & 0.031 & 0.030 & 0.040 & 0.011\\
\hline
\end{tabular}
\end{center}
\end{table}
From these graphs and  Table \ref{tab:4.3}, we observe:
\begin{itemize}

\item
Observation 7: From Figure \ref{fig:4.8} and Table \ref{tab:4.3}, we see that the system dominated by agents of subgroup I has weak "flocking behavior" and less variability. It produces the smallest tail default probability compared to the other two types of networks.

\item
Observation 8: From Figure \ref{fig:4.9} and Table \ref{tab:4.3}, we see that when the agents of subgroup II compose the largest proportion of the system, the system seems to be more stable with the combination of moderate "flocking effect" and variability. Furthermore, the tail default probability has intermediate value when compared to the other two cases.

\item
Observation 9: From Figure \ref{fig:4.10} and Table \ref{tab:4.3}, when agents of subgroup III  dominate the system, the system is relatively unstable. The large $\alpha$'s  produce a strong "flocking behavior" leading the system to overall higher tail default probability.
\end{itemize}

\subsection{Theoretical investigation}
In order to understand theoretically  the behavior of the tail default probability, we need to understand analytically the probability
\begin{equation*}
P\left(\min_{0\leq t \leq T}\frac{1}{N}\sum_{i=1}^N Y_t^{(i)}\leq \eta\right).
\end{equation*}

To simplify some of the computations we assume that $Y_0^{(i)}=0$, for $i=1,...,N$. Furthermore, we want to generalize it to finite agents maintaining different $\sigma_i$'s and $\alpha_i$'s. We divide the N agents into $K < N$ groups with distinct $(\tilde{\alpha_k}, \tilde{\sigma_k})$ for $k = 1,...,K$, $K$ is fixed, and we use $\mathcal{I}_k$ for k'th group where
$$\mathcal{I}_k= \{j\in \{1,...,N\}: (\alpha_j, \sigma_j)=(\tilde{\alpha_k}, \tilde{\sigma_k})\} $$

Namely, within each group the agents are homogeneous. We use $\bar{Y_t}^{k}$ to denote partial empirical averages and $\bar{Y_t}$ for overall empirical average. Also, $\rho_k$ denotes the percentage  of agents that belong to the k'th group. Based on the definition,
\begin{equation*}
 \bar{Y_t}^{k}=\frac{1}{\left|\mathcal{I}_k\right|}\sum_{j\in{\mathcal{I}_k}}Y_t^{(j)}, \text{ where } \; \rho_k=\frac{\left|\mathcal{I}_k\right|}{N} \; \textrm{and }\; k = 1,...,K
\end{equation*}
\begin{equation*}
\bar{Y_t}=\frac{\sum\limits_{i=1}^N Y_t^{(i)}}{N}=\sum_{k=1}^{K} \rho_k\bar{Y_t}^k
\end{equation*}
Therefore, we obtain for the partial average of the kth group:
\begin{equation*}\begin{split}
\bar{Y_t}^{k}&=\frac{1}{\left|\mathcal{I}_k\right|}\sum_{i\in{\mathcal{I}_k}}Y_t^{(i)}
=\frac{1}{\left|\mathcal{I}_k\right|}\sum_{i\in{\mathcal{I}_k}}\left[\alpha_k\int_{0}^{t}(\bar{Y_s}-Y_s^{(i)})ds + \sigma_kW_t^{(i)}\right]\\
&=\alpha_k\int_{0}^{t}\left(\bar{Y_s}-\frac{1}{\left|\mathcal{I}_k\right|}\sum_{i\in{\mathcal{I}_k}}Y_s^{(i)}\right)ds + \sigma_k\frac{1}{\left|\mathcal{I}_k\right|}\sum_{i\in{\mathcal{I}_k}}W_t^{(i)}\\
&=\alpha_k\int_{0}^{t}\left(\bar{Y_s}-\bar{Y_s}^{k}\right)ds +\frac{\sigma_k}{\sqrt{\rho_kN}}\tilde{W_t}^{k}\\
\end{split}
\end{equation*}
where $ \tilde{W_t}^{k}, k = 1,...,K$ are independent standard Brownian motions. Thus,
\begin{equation*}
\begin{split}
d\bar{Y_t}^{k}
&=\alpha_k(\bar{Y_t}-\bar{Y_t}^{k})dt+\frac{\sigma_k}{\sqrt{\rho_kN}}d\tilde{W_t}^{k} \\
&=\left[ \alpha_k(\rho_k-1)\bar{Y_t}^k+\alpha_k\sum_{i=1,i\neq k}^K \rho_i \bar{Y_t^{i}} \right]dt+ \frac{1}{\sqrt{\frac{\rho_kN}{\sigma_k^2}}}d\tilde{W_t}^k
\end{split}
\end{equation*}
Now we write the system in matrix form with column vector $\bar{y_t}=(\bar{Y_t}^{k})_{k=1...K}$ and $\tilde{w_t}=(\tilde{W_t} ^{k})_{k=1...K}$:
\begin{equation*}
d\bar{y_t}=M\bar{y_t}dt+\frac{1}{\sqrt{N}}R^{-1/2}d\tilde{w_t}
\end{equation*}
where $$M_{ij}=-\alpha_i(\delta_{ij}-\rho_j), \; R_{ij}=\rho_i\sigma_i^{-2}\delta_{ij}$$
with
$$\delta_{ij}=\left\{
\begin{aligned}
&1, \; i=j \\
&0, \; i\neq j
\end{aligned}\right. \;\;\;
i,j=1,...,K.$$
By It\^{o} stochastic integration, see for example \cite{KaraztasShreveBook}, the explicit solution to this system is:
\begin{equation*}
\bar{y_t}
=e^{Mt}\bar{y_0}+\frac{1}{\sqrt{N}}\int_{0}^{t}e^{M(t-s)}R^{(-1/2)}d\tilde{w_s}\\
\end{equation*}
Therefore,
\begin{equation*}
\bar{Y_t}=\varrho^T \bar{y}_{t}
=\frac{1}{\sqrt{N}}\varrho^T \int_{0}^{t}  e^{M(t-s)}R^{(-1/2)}d\tilde{w_s}\\
\end{equation*}
and we get that in  distribution,
\[
\bar{Y_t} \sim N \left( \;0\;,\;{\frac{1}{N}V_T^2} \right)
\]
where
\begin{equation}
V_t^2=\int_{0}^{t}\mathcal{\varrho}^Te^{Ms}R^{-1}(e^{Ms})^T\mathcal{\varrho}ds,\; \; \ \varrho=(\rho_k)_{k=1,...,K}.\label{Eq:FormulaForVariance}
\end{equation}

Next, we focus on the ensemble average which reaches the default level at time T. Let $W_t^* \sim  N \left( \;0\;,\; {\frac{1}{N}V^2_t} \right)$. The default probability is:
\begin{equation*} \begin{split}
P\left( \min_{0 \leq t \leq T} \frac{1}{N} \sum_{i=1}^N Y_t^{(i)} \leq \eta \right)
&= P \left( \min_{0 \leq t \leq T} W_t^* \leq \eta \right) = 2P \left( {W_T}^* \leq \eta \right)  \\
&= 2P \left( \frac{1}{\sqrt{N}} V_T \tilde {W} \leq \eta \right)= 2 \Phi \left( \frac{\eta \sqrt{N}}{V_T} \right)
\end{split}
\end{equation*}
Where $\tilde{W} \sim N(0,1)$. Then using Laplace asymptotics, we obtain:

\begin{equation}
\begin{split}
P\left( \min_{0 \leq t \leq T} \frac{1}{N} \sum_{i=1}^N Y_t^{(i)} \leq \eta \right)
\approx 2\exp\left\{-\frac{\eta^2N}{2V_T^2}\right\} \label{Eq:ApproximationDiffAlpha}
\end{split}
\end{equation}
Therefore, we have obtained the following theorem.
\begin{theorem}\label{T:Main}
Consider the full heterogeneous case studied in this section. We  have
\begin{equation} \begin{split}
\lim_{N\to\infty} -\frac{1}{N} \log P \left( \min_{0 \leq t \leq T} \frac{1}{N} \sum_{i=1}^N Y_t^{(i)} \leq \eta \right)
&= \frac{\eta^{2}}{2V_T^2}
\end{split}
\end{equation}
where $V_T^2$ is given by (\ref{Eq:FormulaForVariance}), provided that $V_T^2<\infty$.

\end{theorem}



Let us check the accuracy of the approximation (\ref{Eq:ApproximationDiffAlpha}) choosing the groups $8:1:1$ and $2:5:3$ as in Table \ref{tab:4.1}.  The number of simulations is limited to $M=500$ due to computational resources and  $\{(\alpha,\sigma)_{I},(\alpha,\sigma)_{II},(\alpha,\sigma)_{III}\}=\{(1,2),(10,1),(100,0.5)\}$.

We compute $-\frac{1}{N} \log P(A)$ by simulation and we compare with the asymptotic value $\frac{\eta^2}{2V_T^2}$ as shown as in the following graphs:
\begin{figure}[H]
\centering
\begin{minipage}{.48\textwidth}
\centering
\includegraphics[width=\linewidth, height=5cm]{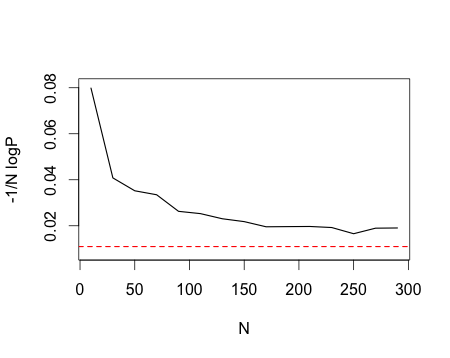}
\caption{Comparison between $-\frac{1}{N} \log P(A)$ and $ \frac{\eta^{2}}{2V_T^2}$ for Group A with Ratio 8:1:1}
\label{fig:*4.7}
\end{minipage}\hfill
\begin{minipage}{.48\textwidth}
\centering
\includegraphics[width=\linewidth, height=5cm]{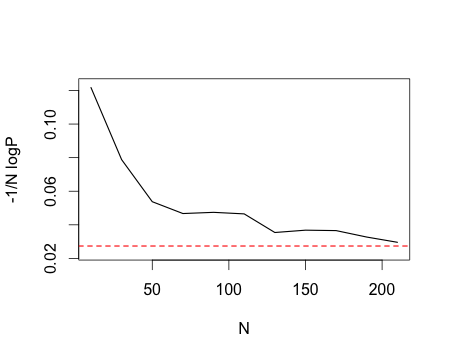}
\caption{Comparison between $-\frac{1}{N} \log P(A)$ and $ \frac{\eta^{2}}{2V_T^2}$ for Group A with Ratio 2:5:3}
\label{fig:*4.8}
\end{minipage}
\end{figure}
Based on the graphs \ref{fig:*4.7} and \ref{fig:*4.8}, we observe that  $-\frac{1}{N} \log P(A)$ converges relatively fast to $\frac{\eta^2}{2V_T^2}$ with increasing $N$, which guarantees the accuracy of (\ref{Eq:ApproximationDiffAlpha}). Of course due to the logarithmic limit of Theorem \ref{T:Main}, one is loosing the prefactor information in this approximation, see Section 4 for more discussion on this issue.

In the case $K=2$, we can obtain an explicit formula for $V_T^2$, which seems to be a much harder task when $K \geq 3$. For this purpose we have Lemma \ref{L:MainK=2}.
\begin{lemma} \label{L:MainK=2}
Consider the setting of the Theorem \ref{T:Main} for K=2. Then, with $\gamma = \alpha_2 \rho_1 + \alpha_1 \rho_2$, we have

\begin{equation*}
\begin{split}
V_T^2 = \frac{\sigma_1^2 \rho_1}{\gamma^2} \left[\alpha_2^2T + \frac{\rho_2^2 \left(\alpha_1 - \alpha_2\right)^2}{2\gamma}\left(1-e^{-2\gamma T}\right) + \frac{2\alpha_2\rho_2 \left(\alpha_1 - \alpha_2\right)}{\gamma}\left(1-e^{-\gamma T}\right)\right] \\
+ \frac{\sigma_2^2 \rho_2}{\gamma^2} \left[\alpha_1^2T + \frac{\rho_1^2 \left(\alpha_1 - \alpha_2\right)^2}{2\gamma}\left(1-e^{-2\gamma T}\right) + \frac{2\alpha_1\rho_1 \left(\alpha_2 - \alpha_1\right)}{\gamma}\left(1-e^{-\gamma T}\right)\right]
\end{split}
\end{equation*}
\end{lemma}
\begin{proof} [Proof of Lemma \ref{L:MainK=2}]
Notice that we can view the matrix M as the infinitesimal generator of a continuous time Markov chain with two states. Let $P(t)=e^{Mt}$ be the transition probability matrix of such a Markov chain. Then, it is known that if we set $c=\alpha_1 \rho_2$ and $d=\alpha_2\rho_1$,
\begin{equation*}
P(t)=
\begin{bmatrix}
\frac{d}{c+d}+\frac{c}{c+d}e^{-\left(c+d\right)t}
&\frac{c}{c+d}-\frac{c}{c+d}e^{-\left(c+d\right)t} \\
\frac{d}{c+d}-\frac{d}{c+d}e^{-\left(c+d\right)t}
&\frac{c}{c+d}+\frac{d}{c+d}e^{-\left(c+d\right)t}
\end{bmatrix}
\end{equation*}
Plugging the expression in the integrand of $V_T^2$ and doing the algebra we get
\begin{small}
\begin{equation*}
\begin{split}
\mathcal{\varrho}^Te^{Ms}R^{-1}(e^{Ms})^T\mathcal{\varrho} = \frac{\sigma_1^2\rho_1}{\gamma^2}\left(\alpha_2 + \rho_2\left(\alpha_1-\alpha_2\right)e^{-\gamma s}\right)^2
+\frac{\sigma_2^2\rho_2}{\gamma^2}\left(\alpha_1 + \rho_1\left(\alpha_2-\alpha_1\right)e^{-\gamma s}\right)^2.
\end{split}
\end{equation*}
\end{small}

By integrating this we get the expression of the lemma.
\end{proof}
In the case $K \geq 3$, an explicit formula seems to be quite hard to derive. We would need the spectral decomposition for $M$. For this reason we derive an approximation to $V^{2}_{T}$ in the case where the $\alpha_{i}'$s are not too different from each other. In particular, we assume that each group differs from another one by the rate $\delta c_i$ where $c_i$ is a bounded real constant, $ c_i \neq c_j$ for $i\neq j$ and $ 0 < \delta \ll 1$. That is:\\
\begin{equation*}
\alpha_k = \bar{\alpha} (1+\delta c_k) \\
\end{equation*}
where $ \bar{\alpha} = \sum\limits_{i=1}^K \rho_i \alpha_i = \sum\limits_{i=1}^N \alpha_i\Big/N $. Let us set
\begin{equation}\label{Eq:ApproximateVariance}
\begin{split}
\hat{V}_{T}^2(\delta)
=T\sum_{i=1}^K \rho_i\sigma_i^{2}+2\delta\left(\frac{1}{\bar{\alpha}}-T-\frac{e^{-\bar{\alpha}T}}{\bar{\alpha}}\right)\sum_{i=1}^K \rho_i c_i \sigma_i^2 \\
+\delta^{2}\left(-\frac{11}{2\bar{\alpha}}+3T+\frac{6+2T\bar{\alpha}}{\bar{\alpha}}e^{-\bar{\alpha}T}-\frac{1}{2\bar{\alpha}}e^{-2\bar{\alpha}T}\right)\sum_{i=1}^K \rho_i c_i^2\sigma_i^2 \\
-2\delta^{2}\left(-\frac{2}{\bar{\alpha}}+T+\frac{2}{\bar{\alpha}}e^{-\bar{\alpha}T}+Te^{-\bar{\alpha}T}\right) \left( \sum_{i=1}^K \rho_i\sigma_i^{2} \sum_{i=1}^K \rho_i c_i^{2} \right)
\end{split}
\end{equation}

Notice that if $\delta$ is small enough so that we can ignore the $O(\delta^{2})$ terms, then large $\bar{\alpha}$ will also imply smaller $\hat{V}_{T}^2(\delta)$. Then, we have the following lemma, whose proof is in the Appendix.
\begin{lemma}\label{L:Main}
Let $V_T^2$ be given by (\ref{Eq:FormulaForVariance}) and ${\hat{V}_{T}^2(\delta)}$ be given by (\ref{Eq:ApproximateVariance}). Then, as $\delta \downarrow 0$, we have the error bound $|V_T^2-\hat{V}_{T}^2(\delta)| \leq K\delta^3+O(\delta^4)$ where $K >0$.
\end{lemma}

Notice that if $\delta=0$ then (\ref{Eq:ApproximateVariance})  gives back the formulas of Section 2. Next, we check numerically the accuracy of the approximation of (\ref{Eq:ApproximateVariance}). We firstly compare the difference of $V_{T}^{2}$ and $\hat{V}_T^2$.
\begin{table}[H]
\caption {Comparison of $V_T^2$ and $\hat{V}_T^2$}
\label{tab:4.2}
\begin{center}
\begin{tabular}
{ p{2cm}|p{2cm}|p{2cm}|p{2cm}|p{2cm} }
 \hline
 \multicolumn{5}{c}{Common parameters:} \\
 \multicolumn{5}{c}{$\delta=0.001$, $c_k=(-60,0,40)$, $\sigma=(5,2,1)$, $\rho=(0.2,0.5,0.3)$, $M=500$, $T=1$}\\
 \hline
 $\bar{\alpha}$ & $\alpha$ & $V_T^2$ & $\hat{V_T}^2$ & $\frac{|\hat{V}_T^2-V_T^2|}{V_T^2}$\\
 \hline
 10   & (9.4,10,10.4)    &7.341& 7.850 & 6.9\%\\
50 & (47,50,52) & 7.401  &7.901 & 6.7\%\\
100 &(94,100,104) & 7.409&  7.907 & 6.7\%\\
 \hline
\end{tabular}
\end{center}
\end{table}

Secondly, we verify the accuracy of (\ref{Eq:ApproximateVariance}) by  comparing $-\frac{1}{N}\log P(A)$ to the approximating value  $\frac{\eta^2}{2\hat{V}_T^2}$  for groups with $\bar{\alpha}=10,50,100$.
\begin{figure}[H]
\centering
\begin{minipage}{.3\textwidth}
\centering
\includegraphics[width=\linewidth]{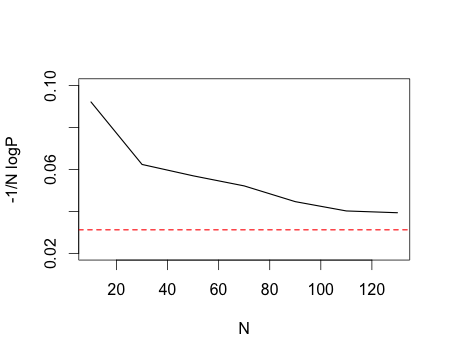}
\caption{Comparison between $-\frac{1}{N} \log P(A)$ and $ \frac{\eta^{2}}{2\hat{V}_T^2}$ under $\bar{\alpha}=10$}
\label{fig:4.11}
\end{minipage}\hfill
\begin{minipage}{.3\textwidth}
\centering
\includegraphics[width=\linewidth]{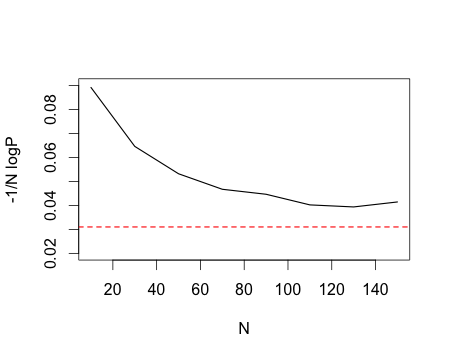}
\caption{Comparison between $-\frac{1}{N} \log P(A)$ and $ \frac{\eta^{2}}{2\hat{V}_T^2}$ under $\bar{\alpha}=50$}
\label{fig:4.12}
\end{minipage}\hfill
\begin{minipage}{.3\textwidth}
\centering
\includegraphics[width=\linewidth]{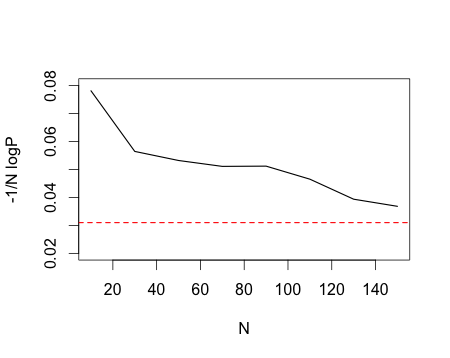}
\caption{Comparison between $-\frac{1}{N} \log P(A)$ and $ \frac{\eta^{2}}{2\hat{V}_T^2}$ under $\bar{\alpha}=100$}
\label{fig:4.13}
\end{minipage}
\end{figure}

From graphs \ref{fig:4.11}, \ref{fig:4.12} and \ref{fig:4.13}, we see that even with a moderate size of $N$, the large deviations approximation of $-\frac{1}{N}\log P(A)$ is reasonably close to $ \frac{\eta^{2}}{2\hat{V}_T^2}$.

\section{Conclusions and future work}\label{S:Conclusion}
This note mainly analyzes the effect of different mean-reversion rates $\alpha_i$'s and volatilities $\sigma_i$'s on systemic risk and flocking behavior. In the model, we use a mean-reversion part to represent the effect of the system to a specific diffusion process and a random part for individual activity. We consider different combinations of ($\alpha_i$,$\sigma_i$) to stretch out heterogeneity effects. Based on our model, we explore numerically the properties of the tail default  probability and compute the exact form of it for the mean behavior. In the specific case where all the $\alpha_i$'s are the same, but the  $\sigma_i$'s are different, we find that $\frac{1}{N}\sum_{i=1}^N \sigma_i^2$ influences the systemic risk positively, but the "flocking effect" resulting from $\alpha_i$'s will also increase the system risk.

In the general case that has different $\alpha_i$'s and $\sigma_i$'s, the situation is far richer. Systemic risk and flocking behavior are significantly affected  by the size of $\alpha_i$'s and $\sigma_i$'s and on how they are combined, see Observations 1-9 in Section 3.1. We also find  that when the majority of agents are of moderate size, in that the corresponding parameters for the coefficients $\alpha_i$ and $\sigma_i$ are neither too large nor too small, the system is more stable. In this case there is a flocking behavior and the tail probability of many agents defaulting is not large.

It is clear that this paper raises open questions. We list some of them below.
\begin{enumerate}
\item{It is clear that the brute force Monte Carlo approximation to $-\frac{1}{N}\log P(A)$ is quite inefficient. As $P(A)$ is a rare event, one would like to develop related accelerated Monte Carlo methods, such as importance sampling.}
\item{Flocking effect and stabilization of the system are issues that need to be better understood and better quantified. Can one quantify, in even more explicit terms than we did in this paper, what is the effect of the interconnections  of the system that the agents constitute on flocking behavior? This paper is a first study towards this direction.}
\item{The model that we studied in this paper is a stylistic model and our goal is to illustrate some of the issues. Analysis of more general models such as the ones appearing in \cite{FouqueIchiba2013,GarnierPapanicoalouYang2013,GieseckeSpiliopoulosSowers2013,SpiliopoulosSowers2015} are of interest.}
\end{enumerate}
We plan to study these questions in future works.

\section{Acknowledgement}
This work was partially supported by the National Science Foundation (NSF) CAREER award DMS 1550918. 

\appendix
\section*{Appendix}
In the appendix we prove Lemma \ref{L:Main}. In order to obtain the expansion of $V_T^2$ with respect to $\delta$, we first obtain an expansion of $\varrho^{T}e^{Mt}$. Based on the form of M, we rewrite it as $M=-\bar{\alpha}\bar{M}-\delta\bar{\alpha}N$. Here $\bar{M}=I-\mu\varrho^{T}$ where $\mu$ is the K-dimensional column vector with each component 1 and $N=c_i(\delta_{ij}-\rho_j), i,j=1,...,K$.
Therefore
\begin{equation*}
e^{Mt}=\sum_{n=0}^{\infty} \frac{(-\bar{\alpha}t)^n}{n!}(\bar{M}+\delta N)^n
\end{equation*}

We focus on Taylor expansion with respect to $\delta\ll 1$. Note the five facts $\sum\limits_{i=1}^K \rho_i=1$, $\bar{M}^n=\bar{M}$, $\bar{M}^T\varrho=0\cdot \mu$, $\bar{M}^TN^T\varrho=N^T\varrho$ and $\bar{M}^T(N^T)^2\varrho=(N^T)^2\varrho$. Below we first present the coefficients of the Taylor expansion of $V_T^2$ with respect to $\delta$ up to second order.

For the zeroth order, i.e. for $O(\delta^0)=O(1)$, the coefficient is:
\begin{equation*}
\varrho^T \sum_{n=0}^\infty\frac{\left(-\bar{\alpha}t\right)^n}{n!}\bar{M}^n
=\sum_{n=0}^\infty\frac{\left(-\bar{\alpha}t\right)^n}{n!} \varrho^T\bar{M}
=\varrho^{T}
\end{equation*}

For the first order, i.e, for  $O(\delta^1)$, the coefficient is:\\
\begin{equation*}
\begin{split}
\varrho^T \sum_{n=1}^\infty\frac{\left(-\bar{\alpha}t\right)^n}{n!}\delta N \bar{M}^{n-1}
&= \sum_{n=1}^\infty\frac{\left(-\bar{\alpha}t\right)^n}{n!}\delta \varrho^T N \bar{M}^{n-1} \\
&=\delta \left[ \frac{-\bar{\alpha}t}{1!} \varrho^{T} N+(e^{-\bar{\alpha}t}-1+\frac{\bar{\alpha}t}{1!}) \varrho^{T} N\bar{M} \right] \\
&=\delta(e^{-\bar{\alpha}t}-1)\varrho^{T}N
\end{split}
\end{equation*}

For the second order, i.e., for  $O(\delta^2)$, the coefficient is:
\begin{equation*}
\varrho^T \left\{ \frac{(-\bar{\alpha}t)^2}{2!}(\delta N)^2+ \sum_{n=3}^\infty  \frac{\left(-\bar{\alpha}t\right)^n}{n!}\left[ \delta^{2} N \bar{M}^{n-2} N  + (\delta N)^2 \bar{M}^{n-2}+ (n-3)\delta^{2} (N \bar{M})^2 \right]  \right\}
\end{equation*}
\begin{equation*}
\begin{aligned}
&=\varrho^T \left\{ \frac{(-\bar{\alpha}t)^2}{2!}(\delta N)^2+ \sum_{n=3}^\infty \frac{\left(-\bar{\alpha}t\right)^n}{n!}\left[ \delta^{2} N^{2}  + (\delta N)^2+ (n-3)\delta^{2} N^2 \right]  \right\}\\
&=\varrho^T \delta^{2} \left[ \sum_{n=2}^\infty \frac{\left(-\bar{\alpha}t\right)^{n}}{n!}(n-1)N^{2} \right] \\
&= \varrho^T \delta^{2} \left[ (-\bar{\alpha}t)\sum_{n=2}^\infty \frac{\left(-\bar{\alpha}t\right)^{n-1}}{(n-1)!}N^{2} - \sum_{n=2}^\infty \frac{\left(-\bar{\alpha}t\right)^{n}}{n!}N^{2}\right]\\
&= \varrho^T \delta^{2} \left[ (-\bar{\alpha}t)(e^{-\bar{\alpha}t}-1)-(e^{-\bar{\alpha}t}-1+\bar{\alpha}t)\right]N^2 \\
&= \delta^{2} \varrho^T  \left(1-e^{-\bar{\alpha}t}-\bar{\alpha}t e^{-\bar{\alpha}t} \right) N^{2}
\end{aligned}
\end{equation*}
Finally, we get the result
\begin{equation*}
\varrho^{T} e^{Mt}= \varrho^{T} + \delta (e^{-\bar{\alpha}t}-1) \varrho^{T} N +\delta^{2} (1-e^{-\bar{\alpha}t}-\bar{\alpha}te^{-\bar{\alpha}t})\varrho^{T} N^{2} + O(\delta^{3})
\end{equation*}
Therefore, we have
\begin{equation*}
\begin{split}
\varrho^{T} e^{Mt} R^{-1} (\varrho^{T} e^{Mt})^{T} = \varrho^{T} R^{-1} \varrho + \delta(e^{-\bar{\alpha}t}-1)\left( \varrho^{T}R^{-1}N^{T}\varrho+\varrho^{T}NR^{-1}\varrho\right) \\
+\delta^2(1-e^{-\bar{\alpha}t}-\bar{\alpha}te^{-\bar{\alpha}t})(\varrho^{T}R^{-1}(N^{T})^{2}\varrho+\varrho^{T}N^{2}R^{-1}\varrho)\\ +\delta^{2}(e^{-\bar{\alpha}t}-1)^{2}\varrho^{T}NR^{-1}N^{T}\varrho+O(\delta^{3}) \footnotemark[3]
\end{split}
\end{equation*}
\footnotetext[3]{Note that we omit $O(\delta^{3})$ terms in this equation.}
Then we can simplify as follows
\begin{equation*}
\begin{split}
\varrho^{T} e^{Mt} R^{-1} (\varrho^{T} e^{Mt})^{T} = \varrho^{T} R^{-1} \varrho + 2\delta(e^{-\bar{\alpha}t}-1)\varrho^{T}NR^{-1}\varrho \\
+2\delta^2(1-e^{-\bar{\alpha}t}-\bar{\alpha}te^{-\bar{\alpha}t})\varrho^{T}N^{2}R^{-1}\varrho\\ +\delta^{2}(e^{-\bar{\alpha}t}-1)^{2}\varrho^{T}NR^{-1}N^{T}\varrho+O(\delta^{3})
\end{split}
\end{equation*}
Noting now the two important facts that $(\varrho^{T}N)_{1j} = \rho_j c_j-\rho_j\sum\limits_{i=1}^{K} \rho_i c_i$ and $(NR^{-1}\varrho)_{j1}= c_j\sigma_j^{2}-c_j\sum\limits_{i=1}^K\rho_i\sigma_i^2$, $j=1,...,K$, we finally obtain:
\begin{equation*}
\varrho^{T}R^{-1}\varrho = \sum_{i=1}^K\rho_i\sigma_i^{2}
\end{equation*}

\begin{equation*}
\varrho^{T}NR^{-1}\varrho= \sum_{i=1}^K\rho_i c_i\sigma_i^{2}-\sum_{i=1}^K\rho_i c_i \sum_{i=1}^K\rho_i\sigma_i^{2}
\end{equation*}

\begin{small}
\begin{equation*}
\begin{split}
\varrho^{T}N^{2}R^{-1}\varrho= \sum_{i=1}^K \rho_i c_i^{2}\sigma_i^{2}-\sum_{i=1}^K\rho_i\sigma_i^{2}\sum_{i=1}^K\rho_i c_i^{2}-\sum_{i=1}^K \rho_i c_i \sum_{i=1}^K\rho_i c_i \sigma_i^{2}+\sum_{i=1}^K\rho_i\sigma_i^{2} \left(\sum_{i=1}^K \rho_i c_i \right)^2
\end{split}
\end{equation*}
\end{small}
\begin{equation*}
\begin{split}
\varrho^{T}NR^{-1}N^{T}\varrho= \sum_{i=1}^K\rho_i c_i^{2}\sigma_i^{2}-2\sum_{i=1}^K\rho_i c_i\sum_{i=1}^K\rho_i c_i\sigma_i^{2}+\sum_{i=1}^K\rho_i\sigma_i^{2} \left(\sum_{i=1}^K\rho_i c_i \right)^2
\end{split}
\end{equation*}
Now, we set
\begin{equation*}
A= \sum_{i=1}^K\rho_i\sigma_i^{2}
\end{equation*}

\begin{equation*}
B=2\delta\left(\sum_{i=1}^K\rho_i c_i\sigma_i^{2}-\sum_{i=1}^K\rho_i c_i \sum_{i=1}^K\rho_i\sigma_i^{2} \right)
\end{equation*}

\begin{small}
\begin{equation*}
C=2\delta^{2} \left[ \sum_{i=1}^K \rho_i c_i^{2}\sigma_i^{2}-\sum_{i=1}^K\rho_i\sigma_i^{2}\sum_{i=1}^K\rho_i c_i^{2}-\sum_{i=1}^K \rho_i c_i \sum_{i=1}^K\rho_i c_i \sigma_i^{2}+\sum_{i=1}^K\rho_i\sigma_i^{2} \left(\sum_{i=1}^K \rho_i c_i \right)^2 \right]
\end{equation*}
\end{small}

\begin{equation*}
D=\delta^{2} \left[\begin{split}
\sum_{i=1}^K\rho_i c_i^{2}\sigma_i^{2}-2\sum_{i=1}^K\rho_i c_i\sum_{i=1}^K\rho_i c_i\sigma_i^{2}+\sum_{i=1}^K\rho_i\sigma_i^{2} \left(\sum_{i=1}^K\rho_i c_i \right)^2
\end{split} \right]
\end{equation*}
Consequently, we get
\begin{equation*}
\int_{0}^{T} \varrho^{T} e^{Ms} R^{-1} (\varrho^{T} e^{Ms})^{T} ds
\end{equation*}
\begin{equation*}
=\int_{0}^{T}\left[A+B(e^{-\bar{\alpha}s}-1)+C(1-e^{-\bar{\alpha}s}-\bar{\alpha}se^{-\bar{\alpha}s})
+D(e^{-2\bar{\alpha}s}+1-2e^{-\bar{\alpha}s})\right] ds
\end{equation*}

Based on this, we get the approximation:
\begin{small}
\begin{equation*}
\begin{split}
V_T^2
&=\int_{0}^{T}\mathcal{\varrho}^Te^{Ms}R^{-1}(e^{Ms})^T\mathcal{\varrho}ds \\
& \approx \left[\frac{1}{2\bar{\alpha}}(2B-4C-3D)+(A-B+C+D)T
+(\frac{2C}{\bar{\alpha}}+CT-\frac{B}{\bar{\alpha}}+\frac{2D}{\bar{\alpha}})e^{-\bar{\alpha}T}-\frac{D}{2\bar{\alpha}}e^{-2\bar{\alpha}T} \right] \\
&=\hat{V}_{T}^2(\delta)
\end{split}
\end{equation*}
\end{small}

Then, with
\begin{equation*}
\alpha_k = \bar{\alpha}(1+\delta c_k)
\end{equation*}
we find that
\begin{equation*}
\bar{\alpha} = \sum_{i=1}^{K} \rho_i\alpha_i = \sum_{i=1}^{K}\rho_i\bar{\alpha}(1+\delta c_i) = \bar{\alpha} + \delta \bar{\alpha} \sum_{i=1}^K \rho_i c_i
\end{equation*}
where $\bar{\alpha}>0, \delta>0$. Equivalently, we have obtained
$\sum_{i=1}^K \rho_i c_i =0$.
 Plugging the latter expression in A,B,C,D and using $\sum_{i=1}^K \rho_i c_i =0$, we conclude the proof of the lemma.
%

\bibliographystyle{apalike}

\end{document}